%% file: main.tex
\documentclass[10pt]{article}
\input{macros}

\usepackage{multicol}
\usepackage{adjustbox}

\newcommand{\GIP}{\mathsf{GIP}}

\newcommand{\prt}{\mathsf{prt}}
\newcommand{\cprt}{\mathsf{cprt}}

\newcommand{\dom}{\operatorname{dom}}
\newcommand{\eps}{\varepsilon}
\newcommand{\bits}{\{0,1\}}
\newcommand{\F}{\mathbb{F}}
\newcommand{\NOF}{\mathrm{NOF}}

\title{Randomized Lifting for One-Way Number-on-Forehead Communication}

\author{Chenyu Wang\\
\emph{\small Department of Information Engineering}\\
\emph{\small The Chinese University of Hong Kong}\\
{\small Email: wc025@ie.cuhk.edu.hk}
}

\begin{document}
\maketitle

\begin{abstract}
We prove a lifting theorem from two-party public-coin one-way communication
to multiparty public-coin one-way number-on-forehead (NOF) communication.
For every fixed $k\ge2$ and prime $q>2k$, there is a generalized inner
product gadget $\GIP_{q,r}^k:(\F_q^r)^k\to\F_q$ with
$r=O_k(q/\log q)$ such that, for every partial Boolean function
$f:D\to\bits$, where $D\subseteq\F_q\times\F_q$,
\[
R_{1/3}^1(f)-O(1)
\le
R_{1/6}^{1,\NOF}\bigl(f\circ\GIP_{q,r}^k\bigr)
\le
R_{1/6}^1(f).
\]
Thus, composition with the gadget preserves one-way randomized communication
complexity up to an additive constant and a change in the error parameter.
The lower bound holds in the general one-way NOF model, where the last
player sees the entire gadget input. This extends the deterministic one-way
NOF lifting theorem of Yang and Zhang to randomized protocols, and extends
the randomized lifting result of Wang and Wu from the conservative model
to the general one-way NOF model.

Our proof introduces a one-way cylinder partition bound that lower bounds
public-coin one-way NOF communication complexity. We show that, for the
lifted function, this bound is at least half the one-way partition bound
of the outer function. The main technical step transfers a dual solution
between the two bounds, using M\"obius inversion and a discrepancy estimate
for generalized inner product to control the loss. Combining this transfer
with the characterization of two-party one-way randomized communication
complexity by the one-way partition bound yields the lifting theorem.
\end{abstract}

\section{Introduction}

Lifting theorems provide a general way to transfer hardness from a simpler computational model to a communication model by composing an outer function with a suitable gadget. Beginning with the simulation theorem of Raz and McKenzie, query-to-communication lifting has developed into a substantial theory, encompassing deterministic and randomized communication and a variety of gadgets \cite{raz1997separation,goos2020query,chattopadhyay2019query}. This framework has found applications to circuit complexity, proof complexity, data structures, and combinatorial optimization \cite{pitassi2017strongly,goos2018deterministic,lovett2020lifting,Collision,mao2025gadgetless}.

The number-on-forehead (NOF) model, introduced by Chandra, Furst, and Lipton \cite{chandra1983multi}, is a central model of multiparty communication complexity in which player \(i\) sees all input blocks except its own. The overlapping information available to the players makes NOF protocols substantially more difficult to analyze than their two-party counterparts. A classical line of work initiated by Babai, Nisan, and Szegedy developed discrepancy methods for this model and established strong lower bounds for explicit functions, including the generalized inner product function \cite{babai1992multiparty}. Thus, generalized inner product has long served as a canonical pseudorandom object for cylinder intersections in multiparty communication.

More recently, Yang and Zhang introduced a lifting framework connecting two-party one-way communication to multiparty one-way NOF communication \cite{yang2025deterministic}. Using a generalized inner product gadget, they proved a deterministic lifting theorem and asked whether an analogous randomized lifting theorem holds. Wang and Wu subsequently established randomized and quantum lifting theorems in the conservative one-way NOF model, in which the last player sees only the gadget output rather than the entire gadget input \cite{wang2026randomized}. Their results leave open the general one-way NOF model, where the last player sees the full gadget input and may exploit this additional information when producing the output.

In this paper, we prove a randomized one-way lifting theorem for the general NOF model using the generalized inner product gadget. For every partial Boolean outer function, the lifted problem preserves its public-coin one-way communication complexity up to an additive constant and a change in the error parameter. Our proof proceeds through a linear programming bound: we introduce a one-way cylinder partition bound and show that it preserves the one-way partition bound of the outer function under composition.

\paragraph{Problem setting:} 
Let \(q\) be prime, let \(k\ge2\), and define
\[
\GIP_{q,r}^{k}(x_1,\ldots,x_k)
\defeq
\sum_{j=1}^{r}
\prod_{i=1}^{k}x_{i,j}
\in\F_q,
\]
where \(x_i\in\F_q^r\) for every \(i\in[k]\).
Given a partial Boolean function $f:D\to\bits$ where $D\subseteq \F_q\times\F_q$, 
we consider the lifted function
\[
F(x_1,\ldots,x_k,z)
=
f\bigl(z,\GIP_{q,r}^{k}(x_1,\ldots,x_k)\bigr).
\]
Denote $\dom(f)$ (resp. $\dom(F)$) as the domain of $f$ (resp. $F$). 

\begin{definition}[Public-coin one-way NOF protocol]
\label{def:one-way-NOF}
Let $\cX=\cX_1\times\cdots\times\cX_k$ and $\cZ$ be finite input
spaces, and let $F:\dom(F)\to\bits$ be a partial function with
$\dom(F)\subseteq\cX\times\cZ$. On input $(x_1,\ldots,x_k,z)$,
player $i\in[k]$ sees $(x_{-i},z)$, where
$x_{-i}=(x_1,\ldots,x_{i-1},x_{i+1},\ldots,x_k)$, and player $k+1$
sees $x=(x_1,\ldots,x_k)$ but not $z$.
The players share a random string $\rho$, independent of the input.
Players $1,\ldots,k$ write binary messages on a shared blackboard in
this order, with
\[
m_i=M_i(x_{-i},z,m_1,\ldots,m_{i-1};\rho),
\qquad i\in[k].
\]
Each player speaks at most once, and empty messages are allowed.
The last player sends no message and outputs
\[
\Pi_\rho(x,z)=g(x,m_1,\ldots,m_k;\rho)\in\bits.
\]
We use the standard binary protocol-tree convention: after fixing
$\rho$, the identity of the next speaker and termination are determined
by the transcript, and the speaker indices along a path are
nondecreasing. Thus message boundaries carry no uncharged information.

The protocol is defined on all of $\cX\times\cZ$, and its communication
cost is
\[
\operatorname{cost}(\Pi)
=\max_{\rho,(x,z)\in\cX\times\cZ}\sum_{i=1}^k |m_i|.
\]
Its worst-case error is at most $\eps$ if
\[
\Pr_\rho[\Pi_\rho(x,z)\ne F(x,z)]\le\eps
\qquad\text{for every }(x,z)\in\dom(F).
\]
No correctness condition is imposed outside $\dom(F)$.
The quantity $R_{\eps}^{1,\NOF}(F)$ is the minimum communication cost
of such a protocol with worst-case error at most $\eps$.
\end{definition}

The basic question is whether composing with
\(\GIP_{q,r}^{k}\) preserves the randomized one-way communication
complexity of the outer function \(f\). Let $R_{\eps}^1(f)$ denote the
public-coin one-way communication complexity of $f$, where Alice receives
$z$, Bob receives $v$, and Alice sends one message to Bob. Let
$R_{\eps}^{1,\NOF}(F)$ denote the corresponding one-way NOF complexity;
a formal definition appears in \Cref{def:one-way-NOF}. Both quantities
use worst-case communication and worst-case error at most $\eps$.
For every $\eps\in(0,1/2)$, we have
\begin{equation}
\label{eq:NOF-upper-bound}
R_{\eps}^{1,\NOF}
\bigl(f\circ\GIP_{q,r}^{k}\bigr)
\le
R_{\eps}^{1}(f).
\end{equation}
Indeed, fix a two-party public-coin protocol for $f$, with Alice's message
$m(z,\rho)$ and Bob's output $b(v,m,\rho)$, where $\rho$ is the public
randomness. In the NOF protocol, player $1$ sees $z$ and writes
$m(z,\rho)$ on the shared blackboard; players $2,\ldots,k$ send empty
messages. The last player sees all of $x=(x_1,\ldots,x_k)$, computes
$v=\GIP_{q,r}^{k}(x)$ locally, and outputs $b(v,m(z,\rho),\rho)$.
For each $(x,z)\in\dom(F)$, this is exactly the original protocol on
$(z,v)\in\dom(f)$, so its error is at most $\eps$ and its communication
cost is unchanged. Minimizing over two-party protocols proves
\eqref{eq:NOF-upper-bound}.
The main content of the paper is the matching lower bound.

\begin{theorem}[Main theorem]
\label{thm:main}
Let $k\ge2$ be fixed, let $q>2k$ be prime, and suppose
\[
r\ge
2^{k-1}
\frac{q\ln3+2\ln q+\ln8}
{\ln(q/(2k))}.
\]
Then, for every partial Boolean function
$f:\F_q\times\F_q\to\bits$,
\[
R_{1/3}^{1}(f)-\log \log 6-2
\leq 
R_{1/6}^{1,\NOF} \bigl(f\circ\GIP_{q,r}^{k}\bigr)
\leq R_{1/6}^{1}(f).
\]
\end{theorem}
Throughout this paper, $\log$ denotes logarithm base 2 and $\ln$ denotes the natural logarithm. 

The upper bound follows from \eqref{eq:NOF-upper-bound}. It remains to prove the lower bound.

\paragraph{Organization.}
We first define the one-way cylinder partition bound for one-way NOF protocols
and derive its dual formulation.
We then prove the transfer theorem between the one-way partition bound
and the one-way cylinder partition bound.
The remaining section establishes the technical ingredients used in
the transfer argument: the multilinear coefficient estimate based on
Möbius inversion, the discrepancy bound for generalized inner product,
and the resulting conditional mixing estimate.

\section{One-way cylinder partition bound}

We use the public-coin version of the one-way NOF model
in \cite{yang2025deterministic}.

\begin{definition}[One-way partition bound]
    \label{def:partition-bd}
    Given $\eps\in (0,1/2)$, the one-way $\eps$-partition bound of $f:\cZ\times \cV\to \bits$, denoted as $\prt_{\eps}^1(f)$, is the optimal value of the following linear program. 
    \begin{subequations}
        \begin{align}
            &\textbf{Primal}\nonumber\\
            \min\quad
            &\sum_{B\subseteq\mathcal Z} w_B\\
            \mathrm{s.t.}\quad
            &\sum_{B\ni z}w_B=1
            &&\forall z\in\mathcal Z,\\
            &\sum_{b\in \{0,1\}}w_{B,v,b}=w_B
            &&\forall B,v,\\
            &\sum_{B\ni z}w_{B,v,f(z,v)}
            \ge 1-\eps
            &&\forall (z,v)\in \dom(f),\\
            &w_B,w_{B,v,b}\ge0.
        \end{align}
    \end{subequations}
    The dual is given by 
    \begin{subequations}
        \begin{align}
            &\textbf{Dual}\nonumber\\
            \max\quad
            &(1-\eps)
            \sum_{(z,v)\in \dom(f)}\mu_{z,v}
            -\sum_{z\in\mathcal Z}\lambda_z \label{prt-D-obj}\\
            \mathrm{s.t.}\quad
            &\sum_{\substack{
                z\in B,\ (z,v)\in \dom(f)\\
                f(z,v)=b}}
            \mu_{z,v}
            \le \Lambda_{B,v}
            &&\forall B,v,b,
            \label{prt-D-mu-Lambda}\\
            &\sum_v\Lambda_{B,v}
            \le
            1+\sum_{z\in B}\lambda_z
            &&\forall B,  \label{prt-D-Lambda-lambda}\\
            &\mu_{z,v}\ge0,\quad
            \Lambda_{B,v}\ge0,\lambda_z\in\R.
        \end{align}
    \end{subequations}

\end{definition}
For convenience, we extend the dual variables by setting $\mu_{z,v}=0$  whenever $(z,v)\notin \dom(f)$.

The following characterization of the one-way randomized communication complexity by the one-way partition bound comes from 
Arunachalam, Doriguello and Jain \cite{arunachalam2023note}. 

\begin{theorem}[\cite{arunachalam2023note}]
\label{thm:partition-characterization}
For any partial function
$f:\cX\times\cY\to\bits$ over finite input spaces and any
$\eps,\delta\in(0,1/2)$,
\[
\log\prt_\eps^1(f)
\le
R_\eps^1(f),
\]
and
\[
R_{\eps+\delta}^1(f)
\le
\log\prt_\eps^1(f)
+
\log\log(1/\delta)+1.
\]
\end{theorem}

Analogous to the two-party one-way model, the one-way cylinder partition bound
(cf. \Cref{def:cylinder-partition}) lower bounds the randomized
communication complexity of one-way NOF protocols. By applying
\Cref{thm:partition-characterization} with
$\eps=\delta=1/6$, it suffices
to prove the following theorem in order to prove the main theorem. 

\begin{theorem}
\label{thm:lift-rcc}
Let $k\ge2$ be fixed, let $q>2k$ be prime, and suppose
\[
r\ge
2^{k-1}
\frac{q\ln3+2\ln q+\ln8}
{\ln(q/(2k))}.
\]
Then, for every partial Boolean function
$f:\F_q\times\F_q\to\bits$,
\[
\cprt_{1/6}^1
\bigl(f\circ\GIP_{q,r}^{k}\bigr)
\ge
\frac12\prt_{1/6}^1(f).
\]
\end{theorem}

The sufficiency of \Cref{thm:lift-rcc} will be shown later.

Denote $\cZ=\cV=\F_q$ and $\cX=(\F_q^r)^k$. Then $\GIP_{q,r}^k$ is a mapping from $\cX$ to $\cV$. Write $x=(x_1,\ldots,x_k)\in\cX$. 
A set $C\subseteq\cX$ is a cylinder in coordinate $i$ if membership in
$C$ does not depend on $x_i$. A cylinder intersection is an intersection
of one cylinder in each coordinate. Equivalently, $C$ is a cylinder
intersection if
\[
\1_C(x_1,\ldots,x_k)
=
\prod_{i=1}^k\phi_i(x_{-i})
\]
for some functions
\[
\phi_i:(\F_q^r)^{k-1}\to\bits.
\]

For $A\subseteq\cX\times\cZ$ and $z\in\cZ$, define
\[
A_z
\defeq
\{x\in\cX:(x,z)\in A\}.
\]
Define
\[
\cA_k
\defeq
\left\{
A\subseteq\cX\times\cZ:
A_z\text{ is a cylinder intersection in }\cX
\text{ for every }z\in\cZ
\right\}.
\]

The following definition is an analogous partition bound for NOF models. 
\begin{definition}[One-way cylinder partition bound]
    \label{def:cylinder-partition}
Given \(\eps\in(0,1/2)\), the one-way cylinder partition bound of \(F:\cX\times \cZ\to \bits\),
denoted by \(\cprt_{\eps}^1(F)\), is the optimal value of the
following linear program. 
\begin{subequations}
\begin{align}
    &\textbf{Primal}\nonumber\\
\min\quad
&\sum_{A\in\mathcal A_k}w_A\\
\mathrm{s.t.}\quad
&\sum_{\substack{A\in\mathcal A_k\\
                  (x,z)\in A}}
  w_A=1
&&\forall (x,z)\in\mathcal X\times\mathcal Z,\\
&\sum_{b\in\{0,1\}}w_{A,x,b}=w_A
&&\forall A\in\mathcal A_k,\ x\in\mathcal X,\\
&\sum_{\substack{A\in\mathcal A_k\\
                  (x,z)\in A}}
  w_{A,x,F(x,z)}
  \ge 1-\eps
&&\forall (x,z)\in \dom(F),\\
&w_A\ge0
&&\forall A\in\mathcal A_k,\\
&w_{A,x,b}\ge0
&&\forall A\in\mathcal A_k,\ x\in\mathcal X,\ b\in\{0,1\}.
\end{align}
\end{subequations}
The dual is the following linear program. 
\begin{subequations}
    \begin{align}
&\textbf{Dual}\nonumber\\
\max\quad
&(1-\eps)
  \sum_{(x,z)\in \dom(F)}\nu_{x,z}
  -\sum_{(x,z)\in\mathcal X\times\mathcal Z}
   \theta_{x,z}
   \label{cprt-D-obj}\\
\mathrm{s.t.}\quad
&\sum_{\substack{z:\,(x,z)\in A\\
                (x,z)\in \dom(F)\\
                F(x,z)=b}}
  \nu_{x,z}
  \le \Gamma_{A,x}
&&\forall A\in\mathcal A_k,\ x\in\mathcal X,\ b\in\{0,1\}, \label{cprt-D-nu-Gamma}\\
&\sum_{x\in\mathcal X}\Gamma_{A,x}
  \le
  1+\sum_{(x,z)\in A}\theta_{x,z}
&&\forall A\in\mathcal A_k,
\label{cprt-D-Gamma-theta}\\
&\nu_{x,z}\ge0, \Gamma_{A,x}\ge0, 
\theta_{x,z}\in\R
&&\forall x\in \cX, z\in \cZ, A\in \cA_k.
\end{align}
\end{subequations}

\end{definition}

For simplicity, we omit the subscripts and use $(\nu,\theta,\Gamma)$ to denote the variables in the dual of the one-way cylinder partition bound. 

The logarithm of the one-way cylinder partition bound lower bounds the
public-coin one-way NOF communication complexity, as follows. 
\begin{theorem} 
    \label{thm:cylinder-partition}
    For any $\eps\in (0,1/2)$ and any partial Boolean function
    $F$ with $\dom(F)\subseteq\cX\times\cZ$,
    \begin{align*}
        R_{\eps}^{1,\NOF} (F)
        \geq \log \cprt_{\eps}^1(F). 
    \end{align*}
\end{theorem}

\begin{proof}
Let $\Pi$ be a public-coin one-way NOF protocol for $F$ with error at
most $\eps$ and communication cost $c$. Fix the public randomness
$\rho$, and let $T_\rho$ be the set of complete transcripts attained
by $\Pi_\rho$ on $\cX\times\cZ$. The binary protocol tree has depth
at most $c$, so $|T_\rho|\le 2^c$.
For $t=(t_1,\ldots,t_k)\in T_\rho$, define its transcript cell
\[
A_{\rho,t}=\{(x,z):\Pi_\rho\text{ has transcript }t\text{ on }(x,z)\}.
\]
These cells partition $\cX\times\cZ$. For each fixed $z$, we have
\[
\1_{(A_{\rho,t})_z}(x)
=\prod_{i=1}^k
\1\{M_i(x_{-i},z,t_1,\ldots,t_{i-1};\rho)=t_i\}.
\]
The $i$th factor is independent of $x_i$, so each slice is a cylinder
intersection and $A_{\rho,t}\in\cA_k$. The last player's output on this
cell is $g_{\rho,t}(x)\defeq g(x,t_1,\ldots,t_k;\rho)$, a function
of $x$ alone. For $x$ incompatible with $t$, extend $g_{\rho,t}(x)$
arbitrarily to a bit.

For every $A\in\cA_k$, $x\in\cX$, and $b\in\bits$, set
\begin{align*}
w_A
&=\E_\rho\left[\sum_{t\in T_\rho}\1\{A_{\rho,t}=A\}\right],\\
w_{A,x,b}
&=\E_\rho\left[\sum_{t\in T_\rho}
  \1\{A_{\rho,t}=A\}\1\{g_{\rho,t}(x)=b\}\right].
\end{align*}
All weights are nonnegative. Since the transcript cells form a
partition for every $\rho$, for every $(x,z)\in\cX\times\cZ$,
\[
\sum_{A\ni(x,z)}w_A=1.
\]
Also, $\sum_{b\in\bits}w_{A,x,b}=w_A$ for every $A,x$, since
$g_{\rho,t}(x)$ is a bit. Finally, for every $(x,z)\in\dom(F)$,
\[
\sum_{A\ni(x,z)}w_{A,x,F(x,z)}
=\Pr_\rho[\Pi_\rho(x,z)=F(x,z)]
\ge 1-\eps.
\]
Thus these weights are feasible for the primal program in
\Cref{def:cylinder-partition}, and their objective value satisfies
\[
\cprt_{\eps}^1(F)
\le\sum_{A\in\cA_k}w_A
=\E_\rho[|T_\rho|]
\le 2^c.
\]
Taking logarithms and minimizing over $\Pi$ proves the theorem.
\end{proof}

Combining \eqref{cprt-D-nu-Gamma} and \eqref{cprt-D-Gamma-theta}, we obtain that for any $A\in \cA_k$, 
\begin{align} \label{cprt-D-reduced-constraint}
    \sum_x \max_{b\in \bits}
    \sum_{\substack{z:\,(x,z)\in A\\
                (x,z)\in \dom(F)\\
                F(x,z)=b}}
    \nu_{x,z}
    \leq 
    1+ \sum_{(x,z)\in A} \theta_{x,z}. 
\end{align}

Now we show that \Cref{thm:lift-rcc} implies \Cref{thm:main}. 

\begin{proof}[Proof of \Cref{thm:main} via \Cref{thm:lift-rcc}]
    By \Cref{thm:lift-rcc}, for any $f$, we have 
    \[
    \cprt_{1/6}^1
    \bigl(f\circ\GIP_{q,r}^{k}\bigr)
    \ge
    \frac12\prt_{1/6}^1(f). 
    \]
    Then 
    \begin{align} \label{eq:cprt-prt}
        \log \cprt_{1/6}^1
    \bigl(f\circ\GIP_{q,r}^{k}\bigr)
    \geq
    \log \prt_{1/6}^1(f)-1. 
    \end{align}
    By \Cref{thm:cylinder-partition}, setting $\eps=1/6$, we have 
    \begin{align} \label{eq:NOF-cprt}
        R_{1/6}^{1,\NOF} (f\circ\GIP_{q,r}^{k})
        \geq \log \cprt_{1/6}^1(f\circ\GIP_{q,r}^{k}). 
    \end{align}
    By \Cref{thm:partition-characterization}, setting $\eps=\delta=1/6$, we have 
    \begin{align} \label{prt-R}
        \log \prt_{1/6}^1(f)
        \geq 
        R_{1/3}^1 (f)-\log \log 6-1. 
    \end{align}
    Combining \Cref{eq:cprt-prt,eq:NOF-cprt,prt-R}, we obtain 
    \begin{align*}
        R_{1/6}^{1,\NOF} (f\circ\GIP_{q,r}^{k})
        \geq R_{1/3}^1 (f)-\log \log 6-2.
    \end{align*}
    Combined with the upper bound \eqref{eq:NOF-upper-bound} 
    \begin{align*}
        R_{1/6}^{1,\NOF} (f\circ\GIP_{q,r}^{k})
        \leq R_{1/6}^1 (f),
    \end{align*}
    it finishes the proof. 
\end{proof}

\section{Proof of the main theorem}

We first analyze the dual of the one-way partition bound. For simplicity, we omit the subscripts and use $(\mu,\lambda,\Lambda)$ to denote the variables in the dual of the one-way partition bound. 
For a feasible $\mu$, define 
\begin{align} \label{def-hv}
    h_v(B)
    \defeq
    \max_{b\in\bits}
    \sum_{\substack{z\in B\\(z,v)\in \dom(f)\\f(z,v)=b}}
    \mu_{z,v}
\end{align}
for $B\subseteq \cZ$. 
Intuitively, $\mu_{z,v}$ is a nonnegative dual weight on the valid input
$(z,v)$. It is not necessarily a probability distribution. 
Moreover, \(h_v(B)\) is the maximum \(\mu\)-mass of inputs in \(B\) consistent with a single output \(b\), when Bob's input is \(v\).

Let $U$ be the uniform distribution on $\cX$. For $v\in\cV$, let
$U_v$ be the uniform distribution on $\{x\in\cX:\GIP_{q,r}^k(x)=v\}$. 
These sets are nonempty when $r\ge1$.

In the dual of \Cref{def:partition-bd}, eliminating all the $\Lambda_{B,v}$ variables gives the following constraints: for any $B\subseteq \cZ$, 
\begin{align} \label{prt-D-reduced-constraint}
    \sum_v \max_{b\in \bits} \sum_{\substack{z\in B\\(z,v)\in \dom(f)\\f(z,v)=b}} \mu_{z,v}
    \leq 
    1+\sum_{z\in B} \lambda_z.
\end{align}
Using the definition in \eqref{def-hv}, it is equivalent to 
\begin{align} \label{eq:h-lambda}
    \sum_vh_v(B)
    \le1+\sum_{z\in B}\lambda_z. 
\end{align}

Denote
\[
\gamma_{k,q,r}
\defeq
\left(\frac{2k}{q}\right)^{r/2^{k-1}}.
\]
Whenever $\gamma_{k,q,r}<\frac1q$, 
denote
\[
\eta_{k,q,r}
\defeq
\frac{2\gamma_{k,q,r}}
{1/q-\gamma_{k,q,r}}.
\]

In order to prove \Cref{thm:lift-rcc}, we prove the following theorem. 

\begin{theorem}
\label{thm:lift-partition}
Suppose $q$ is prime, $k\ge2$, $r\ge1$ such that $\gamma_{k,q,r}<1/q$. 
Then, for any partial function $f:\F_q\times\F_q\to\bits$,
\[
\cprt_{1/6}^1
\bigl(f\circ\GIP_{q,r}^{k}\bigr)
\ge
\frac{\prt_{1/6}^1(f)}
{1+2q3^q\eta_{k,q,r}}.
\]
\end{theorem}
\begin{proof}[Proof of \Cref{thm:lift-rcc} by \Cref{thm:lift-partition}]

If $q>2k$ and
\[
r\ge
2^{k-1}
\frac{q\ln3+2\ln q+\ln8}
{\ln(q/(2k))},
\]
then
\[
\gamma_{k,q,r}
\le
\frac{1}{8q^2 3^q}.
\]
In particular, $q\gamma_{k,q,r}\le1/2$, and hence
\[
\eta_{k,q,r}
=
\frac{2q\gamma_{k,q,r}}
{1-q\gamma_{k,q,r}}
\le
4q\gamma_{k,q,r}.
\]
Therefore,
\[
2q3^q\eta_{k,q,r}
\le
8q^2 3^q\gamma_{k,q,r}
\le1.
\]
Thus \Cref{thm:lift-partition} implies \Cref{thm:lift-rcc}.
\end{proof}

We need the following lemmas. 
\begin{lemma}
\label{lem:dual-mass}
Suppose $\cZ=\F_q$, $\eps\in(0,1/2)$, and
$(\mu,\lambda,\Lambda)$ is feasible for the dual of the one-way
$\eps$-partition bound with nonnegative objective value. Then
\[
\sum_{z,v}\mu_{z,v}
\le
\frac{q}{\eps}.
\]
\end{lemma}

\begin{proof}
    Applying \eqref{prt-D-reduced-constraint} for $B=\{z\}$, we obtain  
\begin{align*}
    \sum_v \mu_{z,v}
    \leq 
    1+\lambda_z.
\end{align*}
Summing over $z\in \cZ=\F_q$, we obtain 
\begin{align*}
    \sum_{z,v} \mu_{z,v}
    \leq 
    q+\sum_{z}\lambda_z.
\end{align*}
Note that the objective value \eqref{prt-D-obj} is 
\begin{align*}
    (1-\eps)\sum_{z,v}\mu_{z,v}
    -\sum_{z}\lambda_z 
    \leq q- \eps \sum_{z,v}\mu_{z,v}. 
\end{align*}
If the objective value \eqref{prt-D-obj} is nonnegative, we have 
\begin{align*}
    \sum_{z,v} \mu_{z,v} \leq \frac{q}{\eps}.
\end{align*}
\end{proof}

For fixed $v$, denote $S_v=\{z:\mu_{z,v}>0\}$ and $s_v=|S_v|$. 
For $A\in \cA_k$ and $x\in \cX$, define $B_A(x)\defeq \left\{z\in \cZ:(x,z)\in A\right\}$. 

\begin{lemma}
\label{lem:cell-threshold-transfer}
Suppose $q$ is prime, $k\ge2$, $r\ge1$ such that $\gamma_{k,q,r}<1/q$. 
Let $(\mu,\lambda,\Lambda)$ be feasible for the dual of the one-way
partition bound. Then, for every $A\in\cA_k$ and $v\in\cV$,
\[
\abs{
\E_{x\sim U_v}h_v(B_A(x))
-
\E_{x\sim U}h_v(B_A(x))
}
\le
\eta_{k,q,r}3^{s_v-1}
\sum_z\mu_{z,v}.
\]
Consequently,
\[
\sum_v
\abs{
\E_{x\sim U_v}h_v(B_A(x))
-
\E_{x\sim U}h_v(B_A(x))
}
\le
\eta_{k,q,r}
\sum_{z,v}3^{s_v-1}\mu_{z,v}.
\]
\end{lemma}

We delay the proof of \Cref{lem:cell-threshold-transfer} to \Cref{sec:proof-of-lemmas}. 
Now we are ready to prove \Cref{thm:lift-partition}. 

\begin{proof}[Proof of \Cref{thm:lift-partition}]
    Fix $\eps=1/6$. Define
    \[
    E_\mu
    \defeq
    \eta_{k,q,r}
    \sum_{v,z}3^{s_v-1}\mu_{z,v}.
    \]
    
    Suppose $(\mu,\lambda,\Lambda)$ is feasible for the dual of the one-way
    partition bound with nonnegative objective value
    \[
    D_{\mu,\lambda}
    \defeq
    (1-\eps)
    \sum_{(z,v)\in \dom(f)}\mu_{z,v}
    -
    \sum_z\lambda_z.
    \]

    First, we show that 
    \begin{align*}
        \cprt_{1/6}^1 (f\circ\GIP_{q,r}^{k})
        \ge
        \frac{D_{\mu,\lambda}}{1+E_{\mu}}.
    \end{align*}

    Let $N_v=|\{x:\GIP(x)=v\}|$. Note that $N_v>0$ for every $v\in\F_q$. 
    Define
    \[
    \nu_{x,z} =\frac{\mu_{z,\GIP(x)}}{N_{\GIP(x)}},
    \qquad
    \theta_{x,z}=\frac{\lambda_z}{|\cX|}.
    \]
    Note that $(\nu, \theta)$ may not be feasible for the dual of the one-way cylinder partition bound, but the objective value \eqref{cprt-D-obj} of $(\nu, \theta)$ equals $D_{\mu,\lambda}$. To be precise, a direct calculation shows that the objective value of $(\nu,\theta)$ is 
    \begin{align*}
        (1-\eps) \sum_{(x,z)\in \dom(F)} \nu_{x,z}
        -\sum_{x,z} \theta_{x,z}
        &=(1-\eps) \sum_{(z,v)\in \dom(f)} \mu_{z,v}
        -\sum_{z} \lambda_z\\
        &=D_{\mu,\lambda}. 
    \end{align*}

    Applying \eqref{eq:h-lambda} for $B_A(x)$, we have 
    \begin{align*}
        \sum_v h_v \left(B_A(x)\right)
        \leq 1+ \sum_{z\in B_A(x)}\lambda_z. 
    \end{align*} 
    Taking the expectation over $x\sim U$, we obtain 
    \begin{align*}
        \sum_v \E_{x\sim U}
        \left[h_v \left(B_A(x)\right)\right]
        \leq 1+ \E_{x\sim U} \left[\sum_{z\in B_A(x)}\lambda_z\right]. 
    \end{align*} 
    
    By \Cref{lem:cell-threshold-transfer}, for any $A\in \cA_k$, 
    \begin{align}
        \sum_{v} \E_{x\sim U_v}[h_v \left(B_A(x)\right)]
        &\leq 
        \sum_{v} \E_{x\sim U}[h_v \left(B_A(x)\right)]+E_{\mu} \nonumber\\
        &=
        \E_{x\sim U} \left[\sum_{v} h_v \left(B_A(x)\right)\right]+E_{\mu}\nonumber\\
        &\overset{\rm (a)}{\leq} \E_{x\sim U} \left[\sum_{z\in B_A(x)}\lambda_z\right]+1
        +E_{\mu} \nonumber
        \\
        &=
        \sum_{(x,z)\in A}\theta_{x,z}+1+E_{\mu}. \label{eq:add-error-cons}
    \end{align}
    where (a) follows from \eqref{eq:h-lambda}. 
    Note that 
    \begin{align*}
        \sum_v\E_{x\sim U_v}h_v(B_A(x))
        =
        \sum_x \max_{b\in \bits}
        \sum_{\substack{z:\,(x,z)\in A\\
                    (x,z)\in \dom(F)\\
                    F(x,z)=b}}
        \nu_{x,z}.
    \end{align*}
    Then \eqref{eq:add-error-cons} becomes 
    \begin{align} \label{eq:first-condition}
        \sum_{x} \max_{b\in \bits} \sum_{\substack{z:\,(x,z)\in A\\
        (x,z)\in \dom(F)\\
        F(x,z)=b}}
        \nu_{x,z}
        \le \sum_{(x,z)\in A} \theta_{x,z}+1+E_{\mu}. 
    \end{align}
    Compared to \eqref{cprt-D-reduced-constraint}, it has an additive error $E_{\mu}$. We now define new variables compensating the error. 
    Define the scaled variables
\[
\widetilde{\nu}_{x,z}
\defeq
\frac{\nu_{x,z}}{1+E_\mu},
\qquad
\widetilde{\theta}_{x,z}
\defeq
\frac{\theta_{x,z}}{1+E_\mu}.
\]
By \eqref{eq:first-condition},
\[
\sum_x
\max_{b\in\bits}
\sum_{\substack{
z:\,(x,z)\in A\\
(x,z)\in \dom(F)\\
F(x,z)=b}}
\widetilde{\nu}_{x,z}
\le
1+
\sum_{(x,z)\in A}\widetilde{\theta}_{x,z}.
\]
Thus $(\widetilde{\nu},\widetilde{\theta})$ is feasible for the reduced
dual obtained after eliminating the variables $\Gamma$.

To obtain a feasible witness for the original dual, define
\[
\widetilde{\Gamma}_{A,x}
\defeq
\max_{b\in\bits}
\sum_{\substack{
z:\,(x,z)\in A\\
(x,z)\in \dom(F)\\
F(x,z)=b}}
\widetilde{\nu}_{x,z}.
\]
Then \eqref{cprt-D-nu-Gamma} holds by definition, while the preceding
inequality gives \eqref{cprt-D-Gamma-theta}. Hence $(\widetilde{\nu},\widetilde{\theta},\widetilde{\Gamma})$ 
is feasible for the full one-way cylinder partition bound dual.

    Then the objective value \eqref{cprt-D-obj} of $(\widetilde{\nu},\widetilde{\theta})$ is 
    \begin{align*}
        (1-\eps) \sum_{(x,z)\in \dom(F)} \widetilde{\nu}_{x,z}
        -\sum_{x,z} \widetilde{\theta}_{x,z}
        &=\frac{D_{\mu,\lambda}}{1+E_{\mu}}. 
    \end{align*}
    Thus, 
    \begin{align*}
        \cprt_{1/6}^1 (F)
        \geq \frac{D_{\mu,\lambda}}{1+E_{\mu}}. 
    \end{align*}
    Note that 
    \begin{align*}
E_\mu
&=
\eta_{k,q,r}
\sum_{v,z}3^{s_v-1}\mu_{z,v}\\
&\overset{\rm (a)}{\leq}
\eta_{k,q,r}3^{q-1}
\sum_{z,v}\mu_{z,v}\\
&\overset{\rm (b)}{\leq}
6q3^{q-1}\eta_{k,q,r}\\
&=
2q3^q\eta_{k,q,r}.
\end{align*}
Here (a) follows from $s_v\le q$, and (b) follows from
\Cref{lem:dual-mass} with $\eps=1/6$.

Therefore,
\[
\cprt_{1/6}^1(F)
\ge
\frac{D_{\mu,\lambda}}
{1+2q3^q\eta_{k,q,r}}.
\]
Optimizing over feasible solutions and applying strong
duality, we obtain
\[
\cprt_{1/6}^1
\bigl(f\circ\GIP_{q,r}^k\bigr)
\ge
\frac{\prt_{1/6}^1(f)}
{1+2q3^q\eta_{k,q,r}}.
\]
This finishes the proof.

\end{proof}

\section{Proof of \Cref{lem:cell-threshold-transfer}}
\label{sec:proof-of-lemmas}

We need the following lemmas to prove \Cref{lem:cell-threshold-transfer}. 

\begin{lemma}[M\"obius inversion \cite{stanley2011enumerative}]
\label{lem:mobius}
Let $\Omega$ be a finite set, and let
$f,g:2^\Omega\to\R$ satisfy
\begin{align} \label{eq:monius-g-f}
    g(S)=\sum_{T\subseteq S}f(T),
    \qquad \forall S\subseteq\Omega.
\end{align}
Then, for every $S\subseteq\Omega$,
\begin{align}
\label{eq:monius-f-g}
    f(S)
    =
    \sum_{T\subseteq S}
    (-1)^{|S|-|T|}g(T).
\end{align}
\end{lemma}

Recall that for fixed $v$, we write
\[
S_v=\{z:\mu_{z,v}>0\},
\qquad
s_v=|S_v|.
\]
For $u\in\bits^{S_v}$ define
\[
g_v(u)
=\max_{b\in\bits}
\sum_{\substack{z\in S_v\\f(z,v)=b}}
\mu_{z,v}u_z.
\]
Then
\[
h_v(B_A(x))
=g_v\bigl((\1_{A_z}(x))_{z\in S_v}\bigr).
\]

Define 
\begin{align*}
    m_v\defeq \sum_z \mu_{z,v}. 
\end{align*}
 
\begin{lemma}
\label{lem:inversion}

Write the unique multilinear expansion
\[
g_v(u)=\sum_{T\subseteq S_v}c_{v,T}\prod_{z\in T}u_z.
\]
Then
\[
\sum_{T\subseteq S_v}|c_{v,T}|
\le3^{s_v-1}m_v.
\]
\end{lemma}

\begin{proof}
If $S_v=\varnothing$, then $g_v\equiv 0$, and the claim is immediate.
Hence, assume that $s_v\geq 1$.

For every $R\subseteq S_v$, let $\1_R\in\bits^{S_v}$ denote the
indicator vector of $R$. Evaluating the multilinear expansion at $\1_R$,
we obtain
\begin{align}
    g_v(\1_R)
    &=
    \sum_{T\subseteq S_v}
    c_{v,T}
    \prod_{z\in T}\1_R(z)\nonumber\\
    &=
    \sum_{T\subseteq R}c_{v,T},
    \label{eq:gv-subset-transform}
\end{align}
because
\[
\prod_{z\in T}\1_R(z)
=
\begin{cases}
1, & T\subseteq R,\\
0, & T\not\subseteq R.
\end{cases}
\]

By \Cref{lem:mobius}, 
\begin{align}
    c_{v,T}
    =
    \sum_{R\subseteq T}
    (-1)^{|T|-|R|}
    g_v(\1_R).
    \label{eq:mobius-cvt}
\end{align}
Since $g_v(u)\geq 0$ for every $u\in\bits^{S_v}$, it follows that
\begin{align*}
    |c_{v,T}|
    &\leq
    \sum_{R\subseteq T}g_v(\1_R).
\end{align*}
Summing over $T\subseteq S_v$, we obtain
\begin{align}
    \sum_{T\subseteq S_v}|c_{v,T}|
    &\leq
    \sum_{T\subseteq S_v}
    \sum_{R\subseteq T}
    g_v(\1_R)\nonumber\\
    &=
    \sum_{R\subseteq S_v}
    2^{s_v-|R|}
    g_v(\1_R).
    \label{eq:coef-first-bound}
\end{align}
Indeed, for each fixed $R\subseteq S_v$, there are exactly
$2^{s_v-|R|}$ sets $T$ satisfying
\[
R\subseteq T\subseteq S_v.
\]

By the definition of $g_v$,
\begin{align*}
    g_v(\1_R)
    &=
    \max_{b\in\bits}
    \sum_{\substack{z\in R\\f(z,v)=b}}
    \mu_{z,v}\\
    &\leq
    \sum_{z\in R}\mu_{z,v}.
\end{align*}
Substituting this into \eqref{eq:coef-first-bound} gives
\begin{align*}
    \sum_{T\subseteq S_v}|c_{v,T}|
    &\leq
    \sum_{R\subseteq S_v}
    2^{s_v-|R|}
    \sum_{z\in R}\mu_{z,v}\\
    &=
    \sum_{z\in S_v}\mu_{z,v}
    \sum_{\substack{R\subseteq S_v\\z\in R}}
    2^{s_v-|R|}.
\end{align*}

Fix $z\in S_v$. Every set $R\subseteq S_v$ containing $z$ can be
written uniquely as
\[
R=\{z\}\cup R',
\qquad
R'\subseteq S_v\setminus\{z\}.
\]
Therefore,
\begin{align*}
    \sum_{\substack{R\subseteq S_v\\z\in R}}
    2^{s_v-|R|}
    &=
    \sum_{R'\subseteq S_v\setminus\{z\}}
    2^{s_v-1-|R'|}\\
    &=
    \sum_{j=0}^{s_v-1}
    \binom{s_v-1}{j}
    2^{s_v-1-j}\\
    &=
    (2+1)^{s_v-1}\\
    &=
    3^{s_v-1}.
\end{align*}
Consequently,
\begin{align*}
    \sum_{T\subseteq S_v}|c_{v,T}|
    &\leq
    3^{s_v-1}
    \sum_{z\in S_v}\mu_{z,v}\\
    &=
    3^{s_v-1}m_v.
\end{align*}
This completes the proof.
\end{proof}

The following absolute discrepancy estimate follows from the
character-sum argument underlying the disperser lemma in
\cite[Lemma 3.2]{yang2025deterministic}. 
\begin{lemma}
\label{lem:absolute-discrepancy}
Suppose $q$ is prime, $k\ge2$, and $r\ge1$. Let $U$ be uniform on $\cX=(\F_q^r)^k$. 
Then, for every cylinder intersection $C\subseteq\cX$ and every
$v\in\F_q$,
\[
\abs{
\Pr[U\in C,\GIP(U)=v]
-\frac{\Pr[U\in C]}q
}
\le
\gamma_{k,q,r}
\defeq
\left(\frac{2k}{q}\right)^{r/2^{k-1}}.
\]
\end{lemma}

\begin{proof}
For $t\in\F_q$, define the additive character $e_q(t)=\exp\left(\frac{2\pi i t}{q}\right)$. 
Then we have the identity
\begin{equation}
\label{eq:elementary-character-average}
\frac1q\sum_{t\in\F_q}e_q(bt)
=
\begin{cases}
1, & b=0,\\
0, & b\ne0.
\end{cases}
\end{equation}

For $m\ge1$ and $a=(a_1,\ldots,a_r)\in\F_q^r$, define 
\begin{align*}
    M_m(a)
    =
    \max_{\psi_1,\ldots,\psi_m}
    \abs{
    \E_{x_1,\ldots,x_m}
    \left[
    \prod_{i=1}^m\psi_i(x_{-i})
    e_q\left(
    \GIP_{q,r}^{m+1}(a,x_1,\dots,x_m)
    \right)
    \right]
    }
\end{align*}
where each $\psi_i:(\F_q^r)^{m-1}\to\bits$ is independent
of $x_i$. Equivalently, the maximization is over all cylinder sets with $m$ factors. For $m=1$, 
\begin{align*}
    M_1(a)
    =
    \max_{\psi_1\in \bits}
    \abs{
    \E_{x_1}
    \left[
        \psi_1(x_{-1})
    e_q\left(
    \GIP_{q,r}^2(a,x_1)
    \right)
    \right]
    }.
\end{align*}
When $\psi_1(x_{-1})=0$, the objective value is $0$. By the nonnegativity, we may choose $\psi_1(x_{-1})=1$. Then by \eqref{eq:elementary-character-average}, we have 
\begin{align*}
    M_1(a)
    &=
    \abs{
    \E_{x_1}
    \left[
    e_q\left(
    \GIP_{q,r}^2(a,x_1)
    \right)
    \right]
    }\\
    &=\1_{\{a=0\}}
\end{align*}

We claim that for $m\ge2$,
\begin{equation}
\label{eq:one-step-cylinder-cs}
M_m(a)^2
\le
\E_D \left[ M_{m-1}(a\odot D)\right],
\end{equation}
where $D=(D_1,\dots,D_r)$ is uniform in $\F_q^r$ and
$a\odot D=(a_1D_1,\ldots,a_rD_r)$.

To prove this, fix $\psi_1,\dots,\psi_m$, write
$X=(x_1,\ldots,x_{m-1})$, and for $u\in \F_q^r$, denote 
\[
J(X,u)
=
\prod_{i=1}^{m-1}\psi_i(X_{-i},u)
e_q\left(
\GIP_{q,r}^{m+1}(a,u,x_1,\dots,x_{m-1})
\right).
\]
Since $|\psi_m|\le1$, by Cauchy--Schwarz inequality, we have 
\begin{align}
\abs{\E_X\left[\psi_m(X)\E_u \left[J(X,u)\right] \right]}^2
&\stackrel{\rm (a)}{\le}
\E_X \left[\abs{\E_u [J(X,u)] }^2\right] \nonumber\\
&\stackrel{\rm (b)}{=}
\E_X \left[\E_u
\left[J(X,u)\right]\cdot \overline{\E_v
\left[J(X,v)\right]} \right]\nonumber\\
&=
\E_{u,v} \left[\E_X
\left[J(X,u)\overline{J(X,v)}\right]\right]\nonumber\\
&\le
\E_{u,v}
\left[\abs{\E_X
\left[J(X,u)\overline{J(X,v)}\right]}\right] \label{eq:psi-J-JJ}
\end{align}
where $\overline{z}$ denote the complex conjugation of the complex number $z$ and $u,v$ are independent and both uniformly distributed over $\F_q^r$. To be precise, (a) follows from Cauchy--Schwarz inequality and (b) follows from the fact that $v$ is an independent copy of $u$. For fixed $u,v$, 
\begin{align*}
    \abs{\E_X
    \left[J(X,u)\overline{J(X,v)}\right]}
    &=\abs{
        \E_X \left[
            \prod_{i=1}^{m-1}\psi_{i}(X_{-i},u) \psi_{i}(X_{-i},v)
            e_q\left(
            \GIP_{q,r}^{m+1}(a,u-v,x_1,\dots,x_{m-1})
            \right)
        \right]
    }\\
    &=\abs{
        \E_X \left[
            \prod_{i=1}^{m-1}\psi_{i}(X_{-i},u) \psi_{i}(X_{-i},v)
            e_q\left(
            \GIP_{q,r}^{m}(a\odot (u-v),x_1,\dots,x_{m-1})
            \right)
        \right]
    }\\
    &\leq M_{m-1}\bigl(a\odot(u-v)\bigr)
\end{align*}
where the last inequality follows from the fact that $\prod_{i=1}^{m-1}\psi_{i}(X_{-i},u) \psi_{i}(X_{-i},v)$ is a cylinder set with $m-1$ factors and the definition of $M_{m-1}$. Combined with \eqref{eq:psi-J-JJ}, we obtain 
\begin{align*}
    \abs{\E_X\left[\psi_m(X)\E_u \left[J(X,u)\right] \right]}^2
    \leq \E_{u,v} \left[M_{m-1}\bigl(a\odot(u-v)\bigr)\right]. 
\end{align*} 
Since $u-v$ is uniform in $\F_q^r$, maximizing over the left hand side yields 
\eqref{eq:one-step-cylinder-cs}.

Iterating \eqref{eq:one-step-cylinder-cs}, using
$(\E Z)^2\le\E[Z^2]$ to move each square inside the
expectation, yields
\begin{align}
M_k(a)^{2^{k-1}}
&\le
\E_{D_2,\ldots,D_k}
M_1(a\odot D_2\odot\cdots\odot D_k) \nonumber\\
&=
\Pr\left[
a\odot D_2\odot\cdots\odot D_k=0
\right] \nonumber\\
&=\prod_{j=1}^r \Pr \left[a_j\prod_{i=2}^k D_{i,j}=0\right] \label{eq:Mk-prod}
\end{align}
where the last line follows from the mutual independence of $D_2,\dots,D_k$. 
Now fix $\alpha\ne0$ and take $a=(\alpha,\ldots,\alpha)$.
For each coordinate $j$, by the union bound, 
\[
\Pr\left[\prod_{i=2}^k D_{i,j}=0\right]
\le
\sum_{i=2}^k\Pr[D_{i,j}=0]
=
\frac{k-1}{q}.
\]
Combined with \eqref{eq:Mk-prod}, we have 
\[
M_k(\alpha,\ldots,\alpha)^{2^{k-1}}
\le
\left(\frac{k-1}{q}\right)^r
\le
\left(\frac{2k}{q}\right)^r.
\]
For arbitrary cylinder set $C\subseteq \cX$ and $\alpha \neq 0$, we then have 
\begin{equation}
\label{eq:cylinder-character}
\abs{
\E_U\left[\1_C(U)e_q(\alpha\GIP(U))\right]
}
\le
\left(\frac{2k}{q}\right)^{r/2^{k-1}}=\gamma_{k,q,r}.
\end{equation}

Finally, \eqref{eq:elementary-character-average} implies 
\begin{align} \label{eq:1-1/q}
    \1_{\{t=v\}}-\frac1q
    =
    \frac1q
    \sum_{\alpha\in\F_q:\alpha\neq0}e_q\bigl(\alpha(t-v)\bigr).
\end{align}
Therefore, 
\begin{align*}
\abs{
\Pr[U\in C,\GIP(U)=v]
-\frac{\Pr[U\in C]}q
}
&= \abs{
    \E \left[\1_C(U) \left(\1_{\{\GIP(U)=v\}}-\frac{1}{q}\right) \right]
}
\\
&\stackrel{\rm (a)}{=}
\abs{
\frac1q\sum_{\alpha\in\F_q:\alpha\neq0}e_q(-\alpha v)
\E_U\left[\1_C(U)e_q(\alpha\GIP(U))\right]
}\\
&\stackrel{\rm (b)}{\leq }
\frac1q\sum_{\alpha\in\F_q:\alpha\neq0}
\abs{
\E_U\left[\1_C(U)e_q(\alpha\GIP(U))\right]
}\\
&\stackrel{\rm (c)}{\leq }
\frac{q-1}{q}\gamma_{k,q,r}\\
&\le
\gamma_{k,q,r},
\end{align*}
where (a) follows from \eqref{eq:1-1/q}, (b) follows from $|e_q(-\alpha v)|= 1$ and (c) follows from \eqref{eq:cylinder-character}. It finishes the proof.

\end{proof}

\begin{proof}[Proof of \Cref{lem:cell-threshold-transfer}]
For every $z\in\cZ$, define 
\[
A_z
\defeq
\{x\in\cX:(x,z)\in A\}.
\]
Since $A\in\cA_k$, each $A_z$ is a cylinder intersection in $\cX$.

Fix $v\in\cV$. Recall that
\[
g_v(u)
=
\max_{b\in\bits}
\sum_{\substack{z\in S_v\\f(z,v)=b}}
\mu_{z,v}u_z,
\qquad
u\in\bits^{S_v}.
\]
For every $x\in\cX$, we have
\[
z\in B_A(x)
\quad\Longleftrightarrow\quad
x\in A_z.
\]
Therefore,
\begin{align}
h_v(B_A(x))
&=
g_v\left(
\bigl(\1_{A_z}(x)\bigr)_{z\in S_v}
\right).
\label{eq:hv-gv-substitution}
\end{align}

Write the unique multilinear expansion
\[
g_v(u)
=
\sum_{T\subseteq S_v}
c_{v,T}\prod_{z\in T}u_z.
\]
Substituting $u_z=\1_{A_z}(x)$ into
\eqref{eq:hv-gv-substitution}, we obtain
\begin{align}
h_v(B_A(x))
&=
\sum_{T\subseteq S_v}
c_{v,T}
\prod_{z\in T}\1_{A_z}(x)\nonumber\\
&=
\sum_{T\subseteq S_v}
c_{v,T}
\1_{C_T}(x),
\label{eq:hv-cylinder-expansion}
\end{align}
where
\[
C_T
\defeq
\bigcap_{z\in T}A_z.
\]
For $T=\varnothing$, we use the convention \(C_\varnothing=\cX\). 

Since every $A_z$ is a cylinder intersection and the intersection of
cylinder intersections is again a cylinder intersection, $C_T$ is a
cylinder intersection for every $T\subseteq S_v$.

We next derive the conditional mixing bound from
\Cref{lem:absolute-discrepancy}. Let $C\subseteq\cX$ be any cylinder
intersection, and write
\[
a
\defeq
\Pr[U\in C,\GIP(U)=v],
\qquad
s
\defeq
\Pr[U\in C],
\qquad
p_v
\defeq
\Pr[\GIP(U)=v].
\]
By \Cref{lem:absolute-discrepancy},
\[
\abs{a-\frac{s}{q}}
\leq
\gamma_{k,q,r}.
\]
Applying the same lemma with $C=\cX$ gives
\[
\abs{p_v-\frac1q}
\leq
\gamma_{k,q,r}.
\]
In particular,
\[
p_v\geq \frac1q-\gamma_{k,q,r}>0.
\]
Since $U_v$ is the distribution of $U$ conditioned on
$\GIP(U)=v$, we have
\begin{align*}
\abs{\Pr[U_v\in C]-\Pr[U\in C]}
&=
\abs{\frac{a}{p_v}-s}\\
&=
\frac{\abs{a-sp_v}}{p_v}\\
&\leq
\frac{
\abs{a-s/q}
+
s\abs{p_v-1/q}
}{
p_v
}\\
&\leq
\frac{2\gamma_{k,q,r}}
{1/q-\gamma_{k,q,r}}\\
&=
\eta_{k,q,r}.
\end{align*}
Equivalently,
\begin{align}
\abs{
\E_{x\sim U_v}[\1_C(x)]
-
\E_{x\sim U}[\1_C(x)]
}
\leq
\eta_{k,q,r}.
\label{eq:conditional-cylinder-mixing}
\end{align}

Using \eqref{eq:hv-cylinder-expansion}, the triangle inequality, and
\eqref{eq:conditional-cylinder-mixing}, we obtain
\begin{align*}
&\abs{
\E_{x\sim U_v}h_v(B_A(x))
-
\E_{x\sim U}h_v(B_A(x))
}\\
&\qquad=
\abs{
\sum_{T\subseteq S_v}
c_{v,T}
\left(
\E_{x\sim U_v}\1_{C_T}(x)
-
\E_{x\sim U}\1_{C_T}(x)
\right)
}\\
&\qquad\leq
\sum_{T\subseteq S_v}
|c_{v,T}|
\abs{
\E_{x\sim U_v}\1_{C_T}(x)
-
\E_{x\sim U}\1_{C_T}(x)
}\\
&\qquad\leq
\eta_{k,q,r}
\sum_{T\subseteq S_v}|c_{v,T}|.
\end{align*}
By \Cref{lem:inversion},
\[
\sum_{T\subseteq S_v}|c_{v,T}|
\leq
3^{s_v-1}m_v,
\]
where
\[
m_v=\sum_z\mu_{z,v}.
\]
Consequently,
\[
\abs{
\E_{x\sim U_v}h_v(B_A(x))
-
\E_{x\sim U}h_v(B_A(x))
}
\leq
\eta_{k,q,r}3^{s_v-1}
\sum_z\mu_{z,v}.
\]

Finally, summing the preceding inequality over $v\in\cV$ gives
\begin{align*}
&\sum_v
\abs{
\E_{x\sim U_v}h_v(B_A(x))
-
\E_{x\sim U}h_v(B_A(x))
}\\
&\qquad\leq
\eta_{k,q,r}
\sum_v3^{s_v-1}
\sum_z\mu_{z,v}\\
&\qquad=
\eta_{k,q,r}
\sum_{z,v}3^{s_v-1}\mu_{z,v}.
\end{align*}
This completes the proof.
\end{proof}

\section*{Acknowledgments}

ChatGPT-5.6 Sol was used to assist with deriving algebraic manipulations in the dual-transfer argument.

\bibliographystyle{alpha}
\bibliography{refs}

\end{document}

%% file: macros.tex
\usepackage{
    amsmath,
    amsthm,
    amssymb,
    mathrsfs,
    mathtools,
    bbm
}

\usepackage[margin=1in]{geometry}

\usepackage{hyperref}
\usepackage[shortlabels]{enumitem}
\usepackage{xcolor}
\usepackage{framed}
\usepackage{tcolorbox}
\usepackage{graphicx}
\usepackage{cleveref}

\definecolor{DarkRed}{rgb}{0.8,0.0,0.0}
\definecolor{DarkGreen}{rgb}{0.0,0.7,0.0}

\hypersetup{
    colorlinks = true,
    linkcolor  = DarkRed,
    citecolor  = DarkGreen,
    urlcolor   = DarkRed
}

\theoremstyle{plain}

\newtheorem{theorem}{Theorem}

\newtheorem{lemma}[theorem]{Lemma}

\theoremstyle{definition}

\newtheorem{definition}{Definition}

\newcommand{\E}{\mathbb{E}}

\newcommand{\R}{\mathbb{R}}

\newcommand{\defeq}{\triangleq}

\newcommand{\abs}[1]{\left\lvert #1 \right\rvert}

\newcommand{\1}{\mathbbm{1}}

\makeatletter
\def\rv{%
    \@ifnextchar\bgroup{\rv@i}{\rv@ii}%
}
\def\rv@i#1{\mathbf{#1}}
\def\rv@ii#1{\mathbf{#1}}
\makeatother

\newcommand{\cA}{\mathcal{A}}

\newcommand{\cV}{\mathcal{V}}

\newcommand{\cX}{\mathcal{X}}
\newcommand{\cY}{\mathcal{Y}}
\newcommand{\cZ}{\mathcal{Z}}
